\documentclass[11pt]{article}
\usepackage{amsfonts,mathrsfs,amssymb,clrscode,multirow,graphicx}
\usepackage{times}
\usepackage[tbtags]{amsmath}
\usepackage{amsthm}
\newcommand{\abs}[1]{\vert #1 \vert}

\newcommand{\Int}{\mathbb Z}
\newcommand{\Nat}{\mathbb N}

\newcommand{\ceil}[1]{\lceil #1 \rceil}
\newcommand{\floor}[1]{\lfloor #1 \rfloor}

\newcommand{\Cc}{\mathscr{C}}
\newcommand{\Pp}{\mathscr{C}}
\newcommand{\poly}{\mathrm{poly}}
\newcommand{\dist}{\mathrm{Ham}}
\newcommand{\edit}{\mathrm{edit}}

\newcommand{\Aa}{\mathscr{A}}
\newcommand{\Dd}{\mathscr{D}}

\newcommand{\RomOmega}{\textrm{$\mathrm{\Omega}$}}
\newcommand{\RomTheta}{\textrm{$\mathrm{\Theta}$}}

\newlength{\actualtopmargin}
\newlength{\actualsidemargin}
 \theoremstyle{plain}
 \newtheorem{theorem}{Theorem}
 \newtheorem{lemma}{Lemma}
 
 \newtheorem{proposition}{Proposition}

 \newtheorem{definition}{Definition}

 \theoremstyle{remark}

 \theoremstyle{plain}
 \newtheorem*{theorem*}{Theorem}
 \newtheorem*{lemma*}{Lemma}
 \newtheorem*{corollary*}{Corollary}
 \newtheorem*{proposition*}{Proposition}
 \newtheorem*{claim*}{Claim}
 
\begin{document}

\sloppy


\title{\Large
 \textbf{
   Property Testing for Cyclic Groups and Beyond
 }\\
}

\author{
 Fran{\c c}ois Le Gall\footnotemark[1]\\
 \and
 Yuichi Yoshida\footnotemark[2]~~\footnotemark[3]
}

\date{}

\maketitle
\thispagestyle{empty}
\pagestyle{plain}
\setcounter{page}{1}

\vspace{-5mm}

\renewcommand{\thefootnote}{\fnsymbol{footnote}}

\begin{center}
{\large
 \footnotemark[1]%
 Department of Computer Science,
 The University of Tokyo\\
legall@is.s.u-tokyo.ac.jp\\
 [2.5mm]
 \footnotemark[2]%
 School of Informatics, Kyoto University\\
yyoshida@kuis.kyoto-u.ac.jp\\
 [2.5mm]
 \footnotemark[3]%
 Preferred Infrastructure, Inc.  
}\\
[8mm]
\end{center}

\begin{abstract}
This paper studies the problem of testing if an input $(\Gamma,\circ)$, where $\Gamma$ is a finite set of unknown size and $\circ$ 
is a binary operation over $\Gamma$ given as an oracle, is close to a specified class of groups.
Friedl et al.~[\emph{Efficient testing of groups}, STOC'05] have constructed an efficient tester using $\poly(\log|\Gamma|)$ queries for the class of 
abelian groups.
We focus in this paper on subclasses of abelian groups, and show that these problems are much harder:
$\RomOmega(|\Gamma|^{1/6})$ queries are necessary to test if the input is close to a cyclic group, 
and $\RomOmega(|\Gamma|^{c})$ queries for some constant $c$
are necessary to test more generally if the input is close to an abelian group 
generated by $k$ elements, for any fixed integer $k\ge 1$. 
We also show that knowledge of the size of the ground set $\Gamma$ helps only for $k=1$, in which case we construct 
an efficient tester using $\poly(\log|\Gamma|)$ queries; 
for any other value $k\ge 2$ the query complexity remains $\RomOmega(|\Gamma|^{c})$.
All our upper and lower bounds hold for both the edit distance and the Hamming distance. These are, to the best of our knowledge,
the first nontrivial lower bounds for such group-theoretic problems in the property testing model and, in particular, they
imply the first exponential separations between the classical and quantum query complexities of testing closeness to 
classes of groups.\end{abstract}

\section{Introduction}

\noindent\textbf{Background:}
Property testing is concerned with the task of deciding whether an object given as
an oracle has (or is close to having) some expected property.
Many properties including algebraic function properties, 
graph properties, computational geometry properties and regular languages have been proved to be efficiently testable.
We refer to, for example, Refs.~\cite{CS06,Kiwi+STACS00,Ron01} for surveys on property testing.
In this paper, we focus on property testing of group-theoretic properties. 
An example is testing whether a function $f\colon G\to H$, where $H$ and $G$ are groups, is a homomorphism.
It is well known that such a test can be done efficiently~\cite{Ben-Or+04,Blum+JCSS93,Shpilka+SICOMP06}.

Another kind of group-theoretic problems deals with the case where the input consists of both a finite set $\Gamma$ and
a binary operation $\circ\colon\Gamma\times\Gamma\to\Gamma$ over it given as an oracle. 
An algorithm testing associativity of the oracle in time $O(\abs{\Gamma}^2)$ 
has been constructed by Rajagopalan and Schulman~\cite{Rajagopalan+SICOMP00}, improving the
straightforward $O(|\Gamma|^3)$-time algorithm. They also showed that $\RomOmega(|\Gamma|^2)$ queries are necessary for this 
task.
Erg\"un et al.~\cite{Ergun+JCSS00} have proposed an algorithm using $\tilde O(\abs{\Gamma})$ queries testing if $\circ$ is \emph{close to} associative, 
and an algorithm using $\tilde O(\abs{\Gamma}^{3/2})$ queries testing if $(\Gamma,\circ)$ is close to 
being both associative and cancellative (i.e., close to the operation of a group). 
They also showed how these results
can be used to check 
whether the input $(\Gamma,\circ)$ is close to an abelian group with $\tilde O(\abs{\Gamma}^{3/2})$ queries.
The notion of closeness discussed in Erg\"un et al.'s work refer to the Hamming distance of multiplication tables, 
i.e., the number of entries in the multiplication table of $(\Gamma,\circ)$ that have to be modified to obtain a binary operation satisfying the prescribed property.

Friedl et al.~\cite{Friedl+STOC05} have shown that, when considering closeness with respect to the edit distance of multiplication tables 
instead of the Hamming distance (i.e., by allowing deletion and insertion
of rows and columns), there exists an algorithm with query and time complexities polynomial in $\log\abs{\Gamma}$ 
that tests whether $(\Gamma,\circ)$ is close to an \emph{abelian group}. An open question is to understand
for which other classes of groups such a test can be done efficiently and, on the other hand, if nontrivial 
lower bounds can be proved for specific classes of groups.

Notice that the algorithm in Ref.~\cite{Friedl+STOC05} has been obtained by first constructing a simple {\em quantum algorithm} that 
tests in $\poly(\log|\Gamma|)$ time if an input $(\Gamma,\circ)$ is close to an abelian group (based on a quantum algorithm 
by Cheung and Mosca~\cite{Cheung+01} computing efficiently the decomposition of a black-box abelian group on a quantum computer), 
and then replacing the quantum part by clever classical tests. One can find this surprising since, classically, computing the decomposition of a black-box 
abelian group is known to be hard \cite{Babai+FOCS84}. This indicates that, in some cases, new ideas in classical property testing can be derived from 
a study of quantum testers.
One can naturally wonder if all efficient quantum algorithms testing closeness to a given class of groups can be converted into efficient classical testers
in a similar way. This question is especially motivated by the fact that Inui and Le Gall~\cite{Inui+09} have constructed a quantum algorithm with query complexity 
polynomial in $\log\abs{\Gamma}$ that tests whether $(\Gamma,\circ)$ is close to a \emph{solvable group} (note that the class of solvable groups includes all abelian groups), and that their techniques can also be used to test efficiently closeness to several {\em subclasses of abelian groups} on a quantum computer, as discussed later.


\noindent\textbf{Our contributions:}
In this paper we investigate these questions by focusing on subclasses of abelian groups. 
We show lower and upper bounds on the 
randomized (i.e., non-quantum) query complexity of testing if the input is close to a cyclic group, 
and more generally on the randomized query complexity of testing if the input is close to an abelian group generated by $k$ elements (i.e., the class
of groups of the form $\Int_{m_1}\times \cdots\times \Int_{m_r}$ where $1\le r\le k$ and $m_1,\ldots,m_r$ are positive integers), 
for any fixed $k\ge 1$ and for both the edit distance and the Hamming distance.
We prove in particular that their complexities vary
dramatically according to the value of $k$ and according to the assumption that the size of $\Gamma$
is known or not.  Table~\ref{table:chart} gives an overview of our results.

\begin{table}[h!]
\renewcommand\arraystretch{1}
\begin{center}
\caption{Lower and upper bounds on the randomized query complexity of testing if $(\Gamma,\circ)$ is close to specific classes of groups. Here
$\epsilon$ denotes the distance parameter, see Section \ref{section:preliminaries} for details.}\label{table:chart}\vspace{3mm}
\begin{tabular}{|l|l|l|l|}
\hline
Target& Distance & Bound & Reference
\tabularnewline\hline
group& edit or Hamming &$\tilde O(|\Gamma|^{3/2})$  & \cite{Ergun+JCSS00} \\
abelian group&  edit& $O(\poly(\epsilon^{-1},\log|\Gamma|))$ &\cite{Friedl+STOC05}\\

cyclic group (size unknown)& edit or Hamming & $\RomOmega(|\Gamma|^{1/6})$&here (Th.~\ref{theorem:unknown-lowerbound-cyclic}) \\
abelian group with $k$ generators & \multirow{2}{*}{edit or Hamming}&\multirow{2}{*}{$\RomOmega(|\Gamma|^{\frac{1}{6}-\frac{4}{6(3k+1)}})$}&
\multirow{2}{*}{here (Th.~\ref{theorem:lowerbound-k})} \\
\hspace{5mm}$[k$: fixed integer $>1]$&&&\\
cyclic group (size known) & edit or Hamming  & $O(\poly(\epsilon^{-1},\log|\Gamma|))$& here (Th.~\ref{theorem:known-upperbound-cyclic})
\tabularnewline\hline
\end{tabular}
\end{center}
\end{table}

Our results show that, with respect to the edit distance, testing closeness to subclasses of abelian groups generally 
requires exponentially more queries than testing closeness to the whole class of abelian groups. We believe that this
puts in perspective Friedl et al.'s work~\cite{Friedl+STOC05} and indicates both the strength and the limitations of 
their results.

The lower bounds we give in Theorems~\ref{theorem:unknown-lowerbound-cyclic} and~\ref{theorem:lowerbound-k} also prove the 
first exponential separations between the quantum and randomized query complexities of testing closeness to a class of groups.
Indeed, the same arguments as in Ref.~\cite{Inui+09} easily show  that, when the edit distance is considered, testing if the input is close to an abelian group generated 
by $k$ elements can be done using $\poly(\epsilon^{-1},\log|\Gamma|)$ queries on a quantum computer, for any value of $k$ and even if $|\Gamma|$ 
is unknown. While this refutes the possibility that 
all efficient quantum algorithms testing closeness to a given class of groups can be converted into efficient classical testers,
this also exhibits a new set of computational problems 
for which quantum computation can be shown to be strictly more efficient than classical computation.


\noindent\textbf{Relation with other works:}
While Ivanyos~\cite{Ivanyos-thesis07} gave heuristic arguments indicating that testing closeness to a group may be hard in general,
we are not aware of any (nontrivial) proven lower bounds 
on the query complexity of testing closeness 
to a group-theoretic property prior to the present work.
Notice that a few strong lower bounds are known for related computational problems, but in different settings.
Babai~\cite{BabaiSTOC91} and Babai and Szemer{\'e}di~\cite{Babai+FOCS84} showed that computing the order of an elementary abelian group 
in the black-box setting requires exponential time --- this task is indeed one of the sometimes called ``abelian obstacles'' to efficient computation in black-box groups. 
Cleve~\cite{Cleve04} also showed strong lower bounds on the query complexity of order finding (in a model based on hidden permutations
rather than on an explicit group-theoretic structure). These results are deeply connected to the subject of the present paper 
and inspired some of our investigations, but do not give bounds in the property testing setting.
The proof techniques we introduce in the present paper are indeed especially tailored for this setting.

\noindent\textbf{Organization of the paper and short description of our techniques:}
Section~\ref{sec:unknown} deals with the case where $|\Gamma|$ is unknown. Our lower bound on the complexity of testing closeness to a cyclic group
(Theorem~\ref{theorem:unknown-lowerbound-cyclic}) 
is proven in a way that can informally be described as follows. We introduce two distributions of inputs:
one consisting of cyclic groups of the form $\Int_{p^2}$, and another consisting of groups of the form~$\Int_p\times \Int_p$, where $p$ is an unknown prime number 
chosen in a large enough set of primes.  
We observe that each group in the latter distribution is far with respect to the edit distance 
(and thus with respect to the Hamming distance too) from any cyclic group.
We then prove that a deterministic algorithm with $o(|\Gamma|^{1/6})$ queries cannot distinguish those 
distributions with high probability. 

Section~\ref{sec:generator} focuses on testing closeness to the class of groups generated by $k>1$ elements, and proves Theorem~\ref{theorem:lowerbound-k} in a similar way. 
For example, when $k>1$ is a fixed odd integer, we introduce two distributions 
consisting of groups isomorphic to 
$G_p=\Int_{p^2}^{(k+1)/2}\times \Int_p^{(k-1)/2}$ and to $H_p=\Int_{p^2}^{(k-1)/2}\times \Int_p^{(k+3)/2}$, respectively. 
Notice that $G_p$ and $H_p$ have the same size.
While $G_p$ is generated by $k$ elements, we observe that $H_p$ is far 
from any group generated by $k$ elements. We then show that any deterministic algorithm with $o(p^{(k-1)/4})=o(|\Gamma|^{1/6-4/6(3k+1)})$ queries cannot distinguish those 
distributions with high probability, even if $p$ (and thus $|\Gamma|$) is known.

Section~\ref{section:orderknown-cyclic} is devoted to constructing an efficient tester for testing closeness to cyclic groups 
when the size $|\Gamma|$ of the ground set is known. The idea behind the tester we propose 
is that, when the size $|\Gamma|$ of the ground set is given, we know that if $(\Gamma,\circ)$ is a cyclic group, then it is isomorphic to the group
$\Int_{|\Gamma|}$. We then take a random element $\gamma$ of $\Gamma$ and
define the map $f\colon\Int_{|\Gamma|}\to\Gamma$ by $f(i)=\gamma^i$ for any $i\in\{0,\ldots,|\Gamma|-1\}$ 
(here the powers are defined carefully to take into consideration the case where the operation $\circ$ is not associative). 
If $(\Gamma,\circ)$ is a cyclic group, then $\gamma$ is a generating element with non negligible probability, in which case the map $f$ will be a group isomorphism. 
Our algorithm will first test if the map $f$ is close to a homomorphism, and then perform additional tests
to check that $f$ behaves correctly on any proper subgroup of $\Int_{|\Gamma|}$. 

\section{Definitions}\label{section:preliminaries}
Let $\Gamma$ be a finite set and $\circ\colon \Gamma\times \Gamma\to \Gamma$ be a binary operation on it. 
Such a couple $(\Gamma,\circ)$ is called a magma. We first define the Hamming distance between two magmas over
the same ground set.
\begin{definition}
Let $(\Gamma,\circ)$ and $(\Gamma,\ast)$ be two magmas over the same ground set $\Gamma$.
The Hamming distance between $\circ$ and $\ast$, denoted $\dist_{\Gamma}(\circ,\ast)$, is
$
\dist_{\Gamma}(\circ,\ast)=|\{(x,y)\in\Gamma\times\Gamma\:|\:x\circ y\neq x\ast y\}|.
$
\end{definition}

We now define the edit distance between tables.
A table of size $k$ is a function $T$ from $\Pi\times \Pi\to \Nat$ where $\Pi$ is an arbitrary subset of $\Nat$ (the set of natural numbers) of size $k$.
We consider three operations to transform a table to another.
An exchange operation replaces, for two elements $a,b\in \Pi$, the value $T(a,b)$ by an arbitrary element of $\Nat$. Its cost is one.
An insert operation on $T$ adds a new element $a\in\Nat\backslash \Pi$: the new table is the extension of $T$ to the domain
$(\Pi\cap\{a\})\times (\Pi\cap\{a\})$, giving a table of size $(k+1)$ where the $2k+1$ new values of the function are set arbitrarily. Its cost is $2k+1$.
A delete operation on $T$ removes an element $a\in \Pi$: the new table is the restriction of $T$ to the domain $(\Pi\backslash\{a\})\times (\Pi\backslash\{a\})$,
giving a table of size $(k-1)$. Its cost is $2k-1$. The edit distance between two tables $T$ and $T'$ is the minimum cost needed to transform
$T$ to $T'$ by the above exchange, insert and delete operations.

A multiplication table for a magma $(\Gamma,\circ)$ is a table $T\colon \Pi\times \Pi\to \Nat$ of size $\abs{\Gamma}$ for which the values
are in one-to-one correspondence with elements in $\Gamma$, i.e.,~there exists a bijection $\sigma\colon \Pi\rightarrow \Gamma$
such that  $T(a,b)=\sigma^{-1}(\sigma(a)\circ\sigma(b))$ for any $a,b\in \Pi$.
We now define the edit distance between two magmas, which will enable us to compare magmas with distinct grounds sets, and especially magmas with ground sets of different sizes. This is the same definition as the one used in Ref.~\cite{Friedl+STOC05}.
\begin{definition}
The edit distance between two magmas $(\Gamma,\circ)$ and $(\Gamma',\ast)$, 
denoted $\edit((\Gamma,\circ),(\Gamma',\ast))$, is the minimum edit distance 
between $T$ and $T'$ where $T$ (resp.~$T'$) runs over all tables corresponding 
to a multiplication table for $(\Gamma,\circ)$ (resp.~$(\Gamma',\ast)$).
\end{definition}

We now explain the concept of distance to a class of groups.
\begin{definition}
Let $\Pp$ be a class of groups and $(\Gamma,\circ)$ be a magma.
We say that $(\Gamma,\circ)$ is $\delta$-far from  $\Pp$ with respect to the Hamming distance  if 
$$
\min_{\substack{\ast\colon\Gamma\times\Gamma\to \Gamma\\ (\Gamma,\ast) \textrm{ is a group in }\Pp}}\dist_{\Gamma}(\circ,\ast)\ge \delta|\Gamma|^2.
$$

\noindent We say that $(\Gamma,\circ)$ is $\delta$-far from  $\Pp$ with respect to the edit distance if
$$
\min_{\substack{(\Gamma',\ast)\\ (\Gamma',\ast) \textrm{ is a group in }\Pp}}\edit((\Gamma,\circ),(\Gamma',\ast))\ge \delta|\Gamma|^2.
$$
\end{definition}
Notice that if a magma $(\Gamma,\circ)$ is $\delta$-far from a class of groups $\Pp$ with respect to the edit distance, then 
$(\Gamma,\circ)$ is $\delta$-far from $\Pp$ with respect to Hamming distance. The converse is obviously false in general.

Since some of our results assume that the size of $\Gamma$ is not known, we cannot suppose that the set $\Gamma$ is given explicitly.
Instead we suppose that an upper bound $q$ of the size of $\Gamma$ is given, and that each element in $\Gamma$ is represented uniquely
by a binary string of length $\ceil{\log_2q}$. One oracle is available that generates a string representing a random element of $\Gamma$,
and another oracle is available that computes a string representing the product of two elements of $\Gamma$. 
We call this representation a binary structure for $(\Gamma,\circ)$. 
This is essentially the same model as the one used in Ref.~\cite{Friedl+STOC05,Inui+09} and in the black-box group literature 
(see, e.g., Ref.~\cite{Babai+FOCS84}). 
The formal definition follows.

\begin{definition}
A binary structure for a magma $(\Gamma,\circ)$ is a triple $(q,O_1,O_2)$ such that 
$q$ is an integer satisfying $q\ge |\Gamma|$, and $O_1,O_2$ are two oracles satisfying the following conditions:
\begin{itemize}
\setlength{\itemsep}{0mm}
\item[(i)]
there exists an injective map $\pi$ from $\Gamma$ to $\Sigma=\{0,1\}^{\ceil{\log_2q}}$;
\item[(ii)]
the oracle $O_1$ chooses an element $x\in \Gamma$ uniformly at random and outputs the (unique) string $z\in \Sigma$ such that $z = \pi(x)$.
\item[(iii)]
on two strings $z_1,z_2$ in the set $\pi(\Gamma)$, the oracle $O_2$ takes the (unique) element $x\in \Gamma$ such that
$x=\pi^{-1}(z_1)\circ \pi^{-1}(z_2)$ and outputs $\pi(x)$. (The action of $O_2$ on strings in $\Sigma\backslash \pi(\Gamma)$ is arbitrary.)

\end{itemize}
\end{definition}

We now give the formal definition of an $\epsilon$-tester. 
\begin{definition}
Let $\Pp$ be a class of groups and let $\epsilon$ be any value such that $0< \epsilon\le 1$.
An $\epsilon$-tester with respect to the edit distance (resp.,~to the Hamming distance) for $\Pp$ is a randomized algorithm $\Aa$ such that, on any binary structure for a magma $(\Gamma,\circ)$,
\begin{itemize}
\setlength{\itemsep}{0mm}
\item[(i)]
$\Aa$ outputs ``PASS" with probability at least $2/3$ if $(\Gamma,\circ)$ satisfies property $\Pp$;
\item[(ii)]
$\Aa$ outputs ``FAIL" with probability at least $2/3$ if $(\Gamma,\circ)$ is $\epsilon$-far from $\Pp$ with respect to the edit distance (resp.,~to the Hamming distance).
\end{itemize}
\end{definition}

\section{A Lower Bound for Testing Cyclic Groups}\label{sec:unknown}
Suppose that we only know that an input instance $(\Gamma,\circ)$ satisfies $|\Gamma| \leq q$, where $q$ is an integer known beforehand.
In this section, we show that any randomized algorithm then requires $\RomOmega(q^{1/6})$ queries to test whether $(\Gamma,\circ)$ is close to the class
of cyclic groups. More precisely, we prove the following result.

\begin{theorem}\label{theorem:unknown-lowerbound-cyclic}
Suppose that 
the size of the ground set is unknown 
and suppose that $\epsilon\le 1/23$.
Then the query complexity of any $\epsilon$-tester for the class of cyclic groups, with respect to 
the Hamming distance or the edit distance, is 
$\RomOmega(q^{\frac{1}{6}})$.
\end{theorem}

Theorem \ref{theorem:unknown-lowerbound-cyclic} is proved using Yao's minimax principle.
Specifically, we introduce two distributions of instances $\Dd_{Y}$ and $\Dd_{N}$ such that every instance in $\Dd_{Y}$ is a cyclic group and every instance in $\Dd_{N}$ is far from the class of cyclic groups.
Then we construct the input distribution $\Dd$ as the distribution that takes an instance from $\Dd_Y$ with probability $1/2$ and from $\Dd_N$ with probability $1/2$.
If we can show that any deterministic algorithm, given $\Dd$ as an input distribution,
requires $\RomOmega(q^{1/6})$ queries to correctly decide whether an input instance is generated by $\Dd_Y$ or $\Dd_N$ with high probability under the input distribution,
we conclude that any randomized algorithm also requires $\RomOmega(q^{1/6})$ queries to test whether an input is close to a cyclic group. 

We now explain in details the construction of the distribution $\Dd$.
Define $q'=\floor{\sqrt{q}}$ and let $R$ be the set of primes in $\{q'/2,\ldots,q'\}$.
From the prime number theorem, we have $|R|= \RomOmega(q'/\log q')$.
We define $\Dd_Y$ as the distribution over binary structures $(q,O_1,O_2)$ for $\Int_{p^2}$ where
the prime $p$ is chosen uniformly at random from $R$ and the injective map $\pi\colon \Int_{p^2} \to \{0,1\}^{\lceil \log_2q\rceil }$ hidden behind the oracles is also chosen uniformly at random.
We define $\Dd_N$ as a distribution over binary structures for $\Int_{p}^2$ in the same manner.
Indeed, the order of any instance generated by those distributions is at most $q$.
Every instance in $\Dd_Y$ is a cyclic group.
From Lemma~\ref{lemma:Ivanyos} below, we know that every instance in $\Dd_N$ is $1/23$-far (with respect to the edit distance, and thus 
with respect to the Hamming distance too) from the class of cyclic groups.  Its proof is included in Appendix.
\begin{lemma}\label{lemma:Ivanyos}
  Let $(G,\circ)$ and $(H,\ast)$ be two nonisomorphic groups. Then 
  $\edit((G,\circ),(H,\ast))\geq \frac{1}{23} \max(|G|^2,|H|^2).$  
\end{lemma}

In order to complete the proof of Theorem \ref{theorem:unknown-lowerbound-cyclic}, it only remains to show that distinguishing 
the two distributions $\Dd_Y$ and $\Dd_N$ is hard. This is the purpose of the following proposition.

\begin{proposition}\label{proposition:hardness}
Any deterministic algorithm that decides with probability larger than $2/3$ whether the input is from the distribution $\Dd_Y$ or from the distribution $\Dd_N$ must use
$\RomOmega(q^{1/6})$ queries.
\end{proposition}
Let us first give a very brief overview of the proof of Proposition \ref{proposition:hardness}. We begin by showing how
the distributions $\Dd_Y$ and $\Dd_N$ described above can equivalently be created by
first taking a random sequence $\ell$ of strings, and then using some constructions 
$\Cc_Y^{\ell}$ and $\Cc^{\ell}_N$, respectively, which are much easier to deal with. 
In particular, the map $\pi$ in the constructions $\Cc_Y^{\ell}$ and $\Cc^{\ell}_N$ is created ``on the fly''
during the computation using the concept of a reduced decision tree. 
We then show (in Lemma \ref{proposition:hardness-simple})  
a $\RomOmega(q^{1/6})$-query lower bound for distinguishing $\Cc_Y^{\ell}$ and~$\Cc^{\ell}_N$.
\begin{proof}[Proof of Proposition \ref{proposition:hardness}]
Let $\Aa$ be a deterministic algorithm with query complexity~$t$. 
We suppose that $t\le q$, otherwise there is nothing to do.
The algorithm~$\Aa$ can be seen as a decision tree of depth $t$.
Each internal node in the decision tree corresponds to a query to either $O_1$ or $O_2$, 
and each edge from such a node corresponds to an answer for it. The queries to $O_2$ are labelled as 
$O_2(s,s')$, for elements $s$ and $s'$ in $\Sigma=\{0,1\}^{\ceil{\log_2 q}}$. Each answer of a query is 
a binary string in $\Sigma$. Each leaf of the decision tree represents a YES or NO decision (deciding whether the input is from   
$\Dd_Y$ or from $\Dd_N$, respectively).

Since we want to prove a lower bound on the query complexity of $\Aa$, we can make freely 
a modification that gives a higher success probability on all inputs (and thus makes the algorithm $\Aa$ more powerful).
We then suppose that, when $\Aa$ goes through 
an edge corresponding to a string already seen during the computation, then $\Aa$ immediately stops and outputs the correct answer.
With this modification, $\Aa$ reaches a leaf if and only if it did not see the same string twice.
We refer to Figure~\ref{fig:dt}(a) for an illustration. 

We first consider the slightly simpler case where the algorithm $\Aa$ only uses strings obtained from previous
oracle calls as the argument of a query to $O_2$.  In other words, we suppose that, 
whenever an internal node $v$ labelled by $O_2(s,s')$ is reached, then both $s$ and $s'$ necessarily label some 
edge in the path from the root of the tree to $v$ (notice that this is the case for the algorithm of Figure~\ref{fig:dt}(a)). 
We will discuss at the end of the proof how to deal with the general case where $\Aa$ can also query
$O_2$ on strings created by itself (e.g., on the all zero string or on strings taken randomly in $\Sigma$).



\begin{figure}[t]
  \begin{center}$
  \begin{array}{ccc}
    \includegraphics[scale=1]{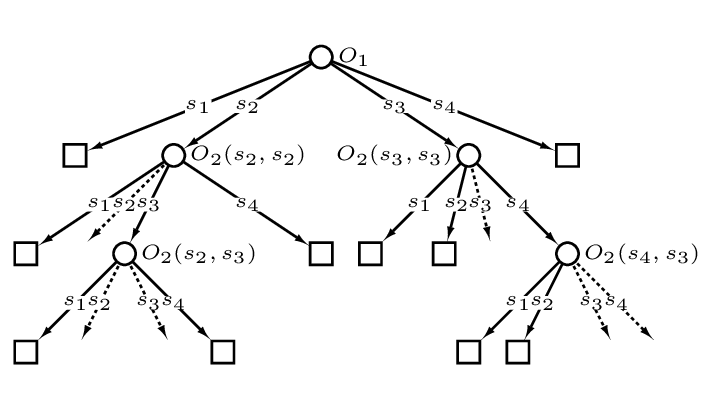}&
      \hspace{2mm}&
     \includegraphics[scale=1]{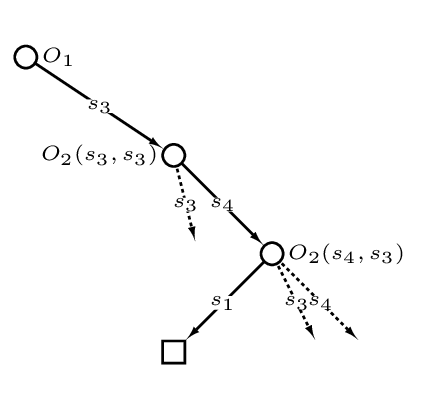}\\      
    {\footnotesize \textrm{(a)}}&& {\footnotesize \textrm{(b)}}\vspace{-5mm}
\end{array}$
       \end{center}
     \caption{
      (a) The decision tree of a deterministic algorithm for $q=4$ and $\Sigma=\{s_1,s_2,s_3,s_4\}$. A dotted arrow means that 
      the computation stops and that the correct answer is systematically output. 
      The leaves are the squared nodes.
      (b) The reduced decision tree associated with the sequence $\ell=(s_3,s_4,s_1,s_2)$. The unseen edges are represented by plain arrows.
      }
    \label{fig:dt}
\end{figure}



Let us fix a sequence $\ell=(\sigma_1,\ldots,\sigma_{\abs{\Sigma}})$ of distinct strings in $\Sigma$. 
Starting from the root $u$ of the decision tree (located at level $i=1$), for each internal node located at level $i\in\{1,\ldots,t\}$, 
we only keep the outgoing branches labelled by strings $\sigma_1,\ldots,\sigma_{i}$, and we call the edge corresponding to 
$\sigma_{i}$ an \textit{unseen edge}  (remember that $t\le q\le |\Sigma|$). 
This construction gives a subtree of the decision tree rooted at $u$ that we call the 
\textit{reduced decision tree associated with $\ell$}. Note that this subtree has exactly one leaf.
See Figure~\ref{fig:dt}(b) for an illustration.

Let us fix $p\in R$ and let $G$ be either $\Int_{p^2}$ or $\Int_p^2$ with the group operation denoted additively.
We now describe a process, invisible to the algorithm $\Aa$, which constructs, using the sequence $\ell$, a map $\pi\colon G\to \Sigma$
defining a binary structure $(q,O_1,O_2)$ for $G$. 
The map $\pi$ is constructed ``on the fly'' during the computation. 
The algorithm starts from the root and follows the computation through the reduced decision tree associated with $\ell$. 
On a node corresponding to a call to $O_1$, 
the oracle $O_1$ chooses a random element $x$ of the group. If this element has not already appeared, 
then $\pi(x)$ is fixed to the string of the unseen edge of this node. The oracle $O_1$ outputs this string to the algorithm $\Aa$, while
$x$ is kept invisible to $\Aa$. 
If the element $x$ has already appeared, then the process immediately stops --- this is coherent with our convention that $\Aa$ stops whenever the same 
string is seen twice. 
On a node corresponding to a call to $O_2(s,s')$, the elements $x$ and $x'$ such that $\pi(x)=s$ and $\pi(x')=s'$ have necessarily been 
already obtained at a previous step from our assumption.
If the element~$x+ x'$ has not already appeared, then $\pi(x+ x')$ is fixed to the string of the 
unseen edge of this node.
Otherwise the process stops.
By repeating this, 
the part of the map $\pi$ related to the computation (i.e., the correspondence between elements and strings for all the elements 
appearing in the computation) is completely defined by $\ell$ and by the elements chosen by the oracle $O_1$.
If necessary, the map $\pi$ can then be completed. 
On the example of Figure~\ref{fig:dt}(b), if the input is $\Int_4=\{0,1,2,3\}$ and $O_1$ chooses the 
element 3, then the path followed is the path starting from the root  labelled by $s_3,s_4,s_1$ which defines 
$\pi(3)=s_3$, $\pi(2)=s_4$, and $\pi(1)=s_1$.

For a fixed sequence $\ell$, let $\Cc^{\ell}_Y$ (resp.~$\Cc^{\ell}_N$) be the ``on the fly'' construction 
for $\Int_{p^2}$ (resp.~$\Int_p^2$) obtained by first choosing $p$ uniformly at random from $R$, and 
then defining $\pi$ while running the algorithm, as detailed above.
The distribution $\Dd_Y$ (resp.~$\Dd_N$) coincides with the distribution that takes a sequence $\ell=(\sigma_1,\ldots,\sigma_{\abs{\Sigma}})$ of $\abs{\Sigma}$
strings in $\Sigma$ uniformly at random without repetition and then create binary structures $(q,O_1,O_2)$ using $\Cc^{\ell}_Y$ (resp.~$\Cc^{\ell}_N$).
Thus, to prove Proposition~\ref{proposition:hardness}, it suffices to use the following lemma.
\begin{lemma}\label{proposition:hardness-simple}
  Let $\ell$ be any fixed sequence of $\abs{\Sigma}$ distinct strings in $\Sigma$.  
  If $\Aa$ decides correctly with probability larger than $2/3$ whether the input has been created using $\Cc_Y^{\ell}$ or using $\Cc_N^{\ell}$, then $t=\RomOmega(q^{1/6})$.
\end{lemma}
\begin{proof}[Proof of Lemma \ref{proposition:hardness-simple}]
Let $v_1,\ldots,v_n$ be the set of nodes in the reduced decision tree associated with $\ell$, 
and let $S\subseteq \{1,\ldots,n\}$ (resp.,~$T\subseteq \{1,\ldots,n\}$) be the set of indexes $i$ such that $v_i$ is a query to $O_2$ (resp.,~to $O_1$).
Notice that 
$|S|+|T|\le t$.
For each index $j \in T$, we set $\alpha_j$ as a random variable representing the element chosen by $O_1$ at node $v_j$.
Here, $\alpha_j \in \Int_{p^2}$ when $\Cc^{\ell}_Y$ generates $\Int_{p^2}$,
and $\alpha_j \in \Int_p^2$ when $\Cc^{\ell}_N$ generates $\Int_p^2$.
Since only additions are allowed as operations on the set $\{\alpha_j\}_{j\in T}$,
the output to a query $v_i$ for $i\in S$ can be expressed as $\pi(a_i)$ where $a_i=\sum_{j\in T} k_j^i\alpha_j$ is a linear combination of the variables in $\{\alpha_j\}_{j\in T}$. Here all coefficients $k_j^i$ are non-negative and at least one coefficient must be positive.

We define the function $a_{ii'}=a_i-a_{i'}=\sum_{j\in T} (k_j^i-k_j^{i'})\alpha_j$ for every $i\neq i'\in S$.
Without loss of generality, we assume that each $a_{ii'}$ is a nonzero polynomial (i.e., there exists at least 
one index $j$ such that $k_j^i\neq k_j^{i'}$). This is because, otherwise, the element (and the string) 
appearing at node $v_i$ is always the same as the element (and the string) appearing at node $v_{i'}$,
and thus one of the two nodes $v_i$ and $v_{i'}$ can be removed from the decision tree. 
For any positive integer $m$, we say that $a_{ii'}$ is \textit{constantly zero modulo $m$} if
$m$ divides $k_j^i-k_j^{i'}$ for all indexes $j\in T$.
We say that a prime $p\in R$ is \textit{good} if there exist $i\neq i'\in S$ such that the function $a_{ii'}$ is constantly zero modulo $p$.
We say that $p\in R$ is \textit{bad} if, for all $i\neq i'\in S$,  the function $a_{ii'}$ is not constantly zero modulo $p$ (as shown later, when $p$ is bad, it is 
difficult to distinguish if the input is $\Int_{p^2}$ or $\Int_p^2$).
We denote by $R_G(\ell)\subseteq R$ the set of good primes.

We first suppose that $|R_G(\ell)|> |R|/6$.
Let $M$ denote the value $\frac{|R|^{1/3}}{\log^{2/3} q'}$.
Assume the existence of a subset $R'_G(\ell) \subseteq R_G(\ell)$ of size  $|R'_G(\ell)| \geq M$  such that there exist
 $i\neq i'\in S$ for which $a_{ii'}$ is constantly zero modulo $p$ for every $p\in R'_G(\ell)$.
Since all $p\in R'_G(\ell)$ are primes, and $a_{ii'}$ is not the zero-polynomial,
$a_{ii'}$ must have a nonzero coefficient divisible by $\prod_{p\in R'_G(\ell)}p$.
To create such a coefficient, 
we must have 
$t \geq \log_2\prod_{p\in R'_G(\ell)}p =\RomOmega(|R'_G(\ell)|\log q')=\RomOmega((|R|\log q')^{1/3}).$ 
Now assume that there exists no such subset $R'_G(\ell)$.
Then, for each $i\neq i'\in S$, 
at most $M$ primes $p$ have the property that $a_{ii'}$ is constantly zero modulo $p$.
This implies that  $|R_G(\ell)|\le M\cdot |S|(|S|-1)/2 \le M\cdot t(t-1)/2 $.
Since $|R_G(\ell)|>|R|/6$,
it follows that $t=\RomOmega((|R|\log q')^{1/3})$.
Thus, for both cases,
we have $t=\RomOmega((|R|\log q')^{1/3})=\RomOmega(q^{1/6})$.

Hereafter we suppose that $|R_G(\ell)|\leq |R|/6$.
Assume that the leaf of the reduced decision tree corresponds to a YES decision.
Recall that, if the computation does not reach the leaf, $\Aa$ always outputs the correct answer. 
From these observations, we give the following upper bound on the overall success probability:
\begin{equation*}
  \frac{r \!+\! (1\!-\!r)(\rho_Y^{\ell}\!\cdot \!1\! +\! (1\!-\!\rho_Y^{\ell})\!\cdot\! 1) }{2}  + 
  \frac{r \!+\! (1\!-\!r)(\rho_N^{\ell}\!\cdot \!1 \!+\! (1\!-\!\rho_N^{\ell})\!\cdot\! 0)}{2}  \\
 \! = \!
  \frac{1\!+\!r\!+\!(1\!-\!r)\rho_N^\ell}{2},
\end{equation*}
where $r=\frac{|R_G(\ell)|}{|R|}$ is the probability of $p$ being good,
and $\rho_Y^{\ell}$ (resp., $\rho_N^{\ell}$) is the probability that $\Aa$ does not reach the leaf conditioned on the event that the instance is from $\Cc_Y^{\ell}$ (resp., from 
$\Cc_N^{\ell}$) and $p$ is a bad prime.
Since $|R_G(\ell)|\leq |R|/6$, the above success probability has upper bound $\frac{7}{12}+\frac{5}{12}\rho_N^{\ell}$.
When the leaf of the reduced decision tree corresponds to a NO decision,
a similar calculation gives that the overall success probability is at most $\frac{7}{12}+\frac{5}{12}\rho_Y^{\ell}$.


We now give an upper bound on $\rho_Y^{\ell}$ and $\rho_N^{\ell}$.
Let us fix $p\in R\setminus R_G(\ell)$.
Since $p$ is bad, each $a_{ii'}$ for $i\neq i'\in S$ is not constantly zero modulo $p$.
When $\Cc^{\ell}_Y$ generates $\Int_{p^2}$,
the probability that $a_{ii'}$ becomes $0$ after substituting values into $\{\alpha_j\}_{j\in T}$ is then exactly $1/p^2$
(since the values of each $\alpha_j$ uniformly distribute over $\Int_{p^2}$ and there is a unique solution in $\Int_{p^2}$ to the equation $a_{ii'}=0$
once all but one values are fixed).
By the union bound,
the probability $\rho_Y^{\ell}$ %
thus satisfies $\rho_Y^{\ell} \leq \frac{|S|(|S|-1)}{2p^2} \leq \frac{t(t-1)}{2p^2} \leq 2\cdot \frac{t(t-1)}{(q')^2}$.
Similarly, when $\Cc^{\ell}_N$ generates $\Int_p^2$, 
the probability that $a_{ii'}$ becomes $0$ after substituting values into $\{\alpha_j\}_{j\in T}$ is also exactly $1/p^2$.
Thus, the probability $\rho_N^{\ell}$ also satisfies $\rho_N^{\ell} \leq \frac{|S|(|S|-1)}{2p^2} \leq \frac{t(t-1)}{2p^2} \leq 2\cdot \frac{t(t-1)}{(q')^2}$.


To achieve overall success probability at least $2/3$, we must have either $\rho^\ell_Y\ge 1/5$ or $\rho^\ell_N\ge 1/5$, and thus
  $t=\RomOmega(q')=\RomOmega(q^{1/2})$.
\end{proof}

Finally, we briefly explain how to deal with the general case where $\Aa$ can make binary strings by itself and use them as arguments to $O_2$.
The difference is that now a string $s$ not seen before can appear as an argument to $O_2$.
Basically, what we need to change is the following two points: 
First, in the  ``on the fly'' construction of $\pi$ from~$\ell$, if such a query appears then an element
$x$ is taken uniformly at random from the set of elements of the input group not already labelled, and the identification $\pi(x)=s$ is done.
Second, in the proof of Lemma \ref{proposition:hardness-simple}, another random variable is introduced to represent the element
associated with $s$. 
With these modifications the same lower bound $t=\RomOmega(q^{1/6})$ holds.

This concludes the proof of Proposition \ref{proposition:hardness}. 
\end{proof}

\section{A Lower Bound for Testing the Number of Generators in a Group}\label{sec:generator}
In this section we show that, even if the size of the ground set $\Gamma$ is known,
it is hard to test whether $(\Gamma,\circ)$ is close to an abelian group generated by $k$ elements for any value $k\geq 2$.
We prove the following theorem using a method similar to the proof of Theorem~\ref{theorem:unknown-lowerbound-cyclic}.
See Appendix for details.
\begin{theorem}\label{theorem:lowerbound-k}
  Let $k\geq 2$ be an integer
and suppose that $\epsilon\le 1/23$. 
Then the query complexity of any  $\epsilon$-tester for the class of abelian groups generated by $k$ elements is 
  \begin{eqnarray*}
    \begin{cases}
      \RomOmega(|\Gamma|^{\frac{1}{6}-\frac{2}{6(3k+2)}}) & \text{if $k$ is even}, \\
      \RomOmega(|\Gamma|^{\frac{1}{6}-\frac{4}{6(3k+1)}}) & \text{if $k$ is odd}.
    \end{cases}
  \end{eqnarray*}
Moreover, these bounds hold with respect to either the Hamming distance or the edit distance, and even when $|\Gamma|$ is known.
\end{theorem}

\section{Testing if the Input is Cyclic when $|\Gamma|$ is Known}\label{section:orderknown-cyclic}
In this section we study the problem of testing, when $|\Gamma|$ is known, if the input $(\Gamma,\circ)$ is a cyclic group or is far from the class of cyclic groups. 
Let us denote $m=|\Gamma|$, and suppose that we also know its factorization $m=p_1^{e_1}\cdots p_r^{e_r}$ 
where the $p_i$'s are distinct primes.
Let $C_m=\{0,\ldots, m-1\}$ be the cyclic group of integers modulo $m$ and, for any $i\in\{1,\ldots,r\}$, 
denote by $C_{m,i}=\{0,\frac{m}{p_i},\ldots,(p_i-1)\frac{m}{p_i}\}$
its subgroup of order $p_i$. The group operation in $C_m$ is denoted additively. 

\begin{figure}[h!]
\hrule
\begin{codebox}
\Procname{Algorithm $\proc{CyclicTest}_\epsilon$} 
\zi \const{input:} a magma $(\Gamma,\circ)$ given as a binary structure $(q,O_1,O_2)$
\zi \hspace{10mm} the size $m=|\Gamma|$ and its factorization $m=p_1^{e_1}\cdots p_r^{e_r}$
\li $\id{decision}\gets \const{FAIL}$; $\id{counter}\gets 0$;
\li \While $\id{decision}=\const{FAIL}$ and $\id{counter}\le d_1=\RomTheta(\log\log m)$ \kw{do}
\li \hspace{3mm}$\id{decision}\gets \const{PASS}$;
\li \hspace{3mm}Take an element $\gamma$ uniformly at random in $\Gamma$;
\li \hspace{3mm}Repeat the following test $d_2=\RomTheta(\epsilon^{-1}\log\log \log m)$ times:
\li \hspace{11mm }take two elements $x,y$ uniformly at random in $C_m$;
\li \hspace{10mm} \If $f_\gamma(x+y)\neq f_\gamma(x)\circ f_\gamma(y)$ \kw{then} $\id{decision}\gets \const{FAIL}$;
\li \hspace{3mm}\For $i\in\{1,\ldots,r\}$ \kw{do}
\li \hspace{10mm} take two arbitrary distinct elements $x,y$ in $C_{m,i}$; 
\li \hspace{10mm} take $d_3=\RomTheta(\log\log\log m)$ elements $u_1,\ldots,u_{d_3}$  at random in $C_m$;
\li \hspace{10mm} \If there exists $j\!\in\!\{1,\ldots,d_3\}$ such that $f_\gamma(x+u_j)= f_\gamma(y+u_j)$ 
\li \hspace{20mm}\kw{then} $\id{decision}\gets\const{FAIL}$;\End
\li \hspace{3mm} $\id{counter}\gets\id{counter}+1$; 
\li \kw{output} $\id{decision}$; 
\end{codebox}
\hrule
\caption{Algorithm $\proc{CyclicTest}_\epsilon$. }\label{figure:CyclicTest}
\end{figure}

For any $\gamma\in \Gamma$, we now define a map $f_\gamma: C_m\to \Gamma$ such that $f_\gamma(a)$ represents the $a$-th power 
of $\gamma$. Since the case where $\circ$ is not associative has to be
taken in consideration and since we want to evaluate efficiently $f$, this map is defined using the following rules.
$$
\left\{\begin{array}{ll}
f_\gamma(1)=\gamma&\\
f_\gamma(a)=\gamma\circ f(a-1)&\textrm{ if $2\le a\le m-1$ and $a$ is odd}\\
f_\gamma(a)= f_\gamma(a/2)\circ f_\gamma(a/2)&\textrm{ if $2\le a\le m-1$ and $a$ is even}\\
f_\gamma(0)=\gamma\circ f(m-1)
\end{array}\right.
$$
The value of $f_\gamma(a)$ can then be computed with $O(\log m)$ uses of the operation~$\circ$.
Notice that if $(\Gamma,\circ)$ is a group, then 
$f_\gamma(a)=\gamma^a$ for any $a\in\{0,\ldots,m-1\}$.

For any $\epsilon>0$, our $\epsilon$-tester for cyclic groups is denoted 
$\proc{CyclicTest}_\epsilon$ and is described in Figure~\ref{figure:CyclicTest}. 
The input $(\Gamma,\circ)$ is given as a binary structure $(q,O_1,O_2)$ with $q\ge m$.
In the description of Figure~\ref{figure:CyclicTest}, operations in $(\Gamma,\circ)$, such as taking a random element or computing the
product of two elements, are implicitly performed by using the oracles $O_1$ and $O_2$. 
The correctness of this algorithm and upper bounds on its complexity are shown in the following theorem.
A proof is given in Appendix.
\begin{theorem}\label{theorem:known-upperbound-cyclic}
For any value $\epsilon>0$, Algorithm $\proc{CyclicTest}_\epsilon$ is an $\epsilon$-tester for cyclic groups with respect to both the edit distance and the Hamming distance.
Its query and time complexities are 
$
O\left((\log m + \frac{\log\log m}{\epsilon})\cdot \log q\cdot \log\log\log m\right).
$
\end{theorem}


\section*{Acknowledgments}
The authors are grateful to G{\'a}bor Ivanyos for communicating to them Lemma~\ref{lemma:Ivanyos} and an outline of its proof.
Part of this work was conducted while YY was visiting Rutgers University.
FLG
acknowledges support from the JSPS,
under the grant-in-aid for research activity start-up No.~22800006.

\newpage 
\setcounter{section}{6}

\section*{Appendix}

\subsection*{A. Proof of Lemma \ref{lemma:Ivanyos}}
The idea of this proof has been communicated to us by Ivanyos \cite{Ivanyos-personal10}.
Work on other aspects of the distance between non-isomorphic groups has subsequently been the subject of a joint paper~\cite{Ivanyos+11}.

We will use the following lemma, which is a weak version of Corollary 1 in Ref.~\cite{Ivanyos+11}.
\begin{lemma}\label{lemma:Ivanyos+}
  Let $(G,\circ)$ and $(H,\ast)$ be two groups such that $|G|\le |H|$.
  If $(G,\circ)$ is not isomorphic to a subgroup of $(H,\ast)$, then
  $$
\Pr_{x,y\in G}[\gamma(x\circ y)=\gamma(x)\ast \gamma(y)]  \le \frac{7}{9}|G|^2
$$
for any injective map $\gamma:G\to H$.
\end{lemma}

We now present our proof of Lemma \ref{lemma:Ivanyos}.
\begin{proof}[Proof of Lemma \ref{lemma:Ivanyos}]
  We assume without loss of generality that $|G|\leq |H|$ and prove the lemma by contraposition.
  Namely, we show that $G$ and $H$ are isomorphic if $\edit((G,\circ),(H,\ast))< |H|^2/23$.
  
    Suppose that $\edit((G,\circ),(H,\ast))< \delta |H|^2$, where $\delta \le 1/23$.
  Let $T_G \colon \Pi_G\times \Pi_G \to \Nat$ and $T_H \colon \Pi_H \times \Pi_H \to \Nat$ be multiplication tables of $G$ and $H$, respectively, such that the edit distance between $T_G$ and $T_H$ is at most $\delta|H|^2$.
  Here, $\Pi_G$ and $\Pi_H$ are subsets of $\Nat$ of size $|G|$ and $|H|$, respectively.
  Let $\sigma_G\colon \Pi_G\to G$ and $\sigma_H\colon \Pi_H\to H$ be the bijections associated with $T_G$ and $T_H$, respectively.

  First notice that $|G| \geq (1-\delta)|H|$. Otherwise, at least $\delta|H|$ elements should be added to $T_G$ to obtain the table $T_H$, which 
  would cost at least  $$\sum_{i=1}^{\delta|H|}(2|H|-2i+1)=2\delta|H|^2-\delta|H|(\delta|H|+1)+\delta|H| = \delta(2-\delta)|H|^2 > \delta|H|^2$$ operations.

  We now consider the transition from $T_G$ to $T_H$ through the process of computing the edit distance.
  Observe that the number of removed elements through the transition is at most $\delta|G|$, otherwise it would cost more than
  \begin{eqnarray*}
  \sum_{i=1}^{\delta|G|}(2|G|-2i+1)&=&2\delta|G|^2-\delta|G|(\delta|G|+1)+\delta|G|\\
  &=&\delta(2-\delta)|G|^2 \geq \delta(2-\delta)(1-\delta)^2|H|^2 > \delta|H|^2
  \end{eqnarray*}
  operations. Let $S\subseteq \Pi_G$ be the set of elements that are not removed in the transition and define
  $U=\{\sigma_G(s)|s\in S\}\subseteq G$.
  From the argument above, we have $|U|\geq (1-\delta)|G|$. 
  
  We define a map $f\colon G\to H$ as follows.
  For $x\in U$, $f(x)=\sigma_H(\sigma_G^{-1}(x))$.
  For $x\not \in U$, we choose $f(x)$ so that $f(x)$ becomes an injective map
  (this is possible since $|G|\leq |H|$).
  Suppose that, for two elements $x,y\in U$,
  the element $x\circ y$ is in $U$.
  Also, suppose that the value $T_G(\sigma_G^{-1}(x),\sigma_G^{-1}(y))$ was not modified in the transition, i.e.,
  $T_G(\sigma_G^{-1}(x),\sigma_G^{-1}(y))=T_H(\sigma_G^{-1}(x),\sigma_G^{-1}(y))$.
  In this case, 
  \begin{eqnarray*}
    \sigma_H^{-1}(f(x) \ast f(y)) &=& T_H(\sigma_H^{-1}(f(x)),\sigma_H^{-1}(f(y)))\\
    &=& T_H(\sigma_G^{-1}(x),\sigma_G^{-1}(y)) \\
    &=& T_G(\sigma_G^{-1}(x),\sigma_G^{-1}(y)) \\
    &=& \sigma_G^{-1}(x\circ y).
  \end{eqnarray*} 
  Thus, we have $f(x) \ast f(y) = \sigma_H(\sigma_G^{-1}(x\circ y)) = f(x \circ y)$.
  Since the number of exchange operations done to the table $T_G$ is at most $\delta|H|^2 \leq \delta|G|^2/(1-\delta)^2$,
  by the union bound we obtain
  $$\Pr_{x,y\in G}[f(x\circ y)=f(x)\ast f(y)] \geq 1-3\delta-\delta/(1-\delta)^2 \geq 1-5\delta.$$
  Thus, since $5\delta<2/9$, Lemma~\ref{lemma:Ivanyos+} implies that the group $(G,\circ)$ is isomorphic to a subgroup of $(H,\ast)$.
  If $(G,\circ)$ is isomorphic to a proper subgroup of $(H,\ast)$, then $|G|\le |H|/2$,
  which contradicts the fact that  $|G| \geq (1-\delta)|H|$.
 Thus, $(G,\circ)$ is indeed isomorphic to $(H,\ast)$.
\end{proof}
\subsection*{B. Proof of Theorem~\ref{theorem:lowerbound-k}}
To show the lower bound, we use Yao's minimax principle as in the proof of Theorem~\ref{theorem:unknown-lowerbound-cyclic}.
We introduce two distributions $\Dd_Y$ and $\Dd_N$ such that every instance in $\Dd_Y$ is generated by $k$ elements while every instance in $\Dd_N$ is far from abelian groups generated by $k$ elements.
Moreover, all instances in $\Dd_Y$ and $\Dd_N$ have the same order. 
Then we construct the input distribution $\Dd$ as the distribution that takes an instance from $\Dd_Y$ with probability $1/2$ and from $\Dd_N$ with probability $1/2$.
By showing that any deterministic algorithm requires many queries to distinguish them, we obtain 
the desired result.

We first consider the case where $k$ is even.
Let $r\geq 2$ be a fixed integer and denote $k=2r-2$.
For any fixed (and known) prime $p$,
we define $\Dd_Y$ as the distribution over binary structures for the group $\Int_{p^2}^r\times \Int_p^{r-2}$
where the injective map $\pi$ hidden behind the group oracles is chosen uniformly at random.
We define $\Dd_N$ as the uniform distribution over binary structures for $\Int_{p^2}^{r-1}\times \Int_p^r$ in the same manner.
The order of every instance in $\Dd_Y$ and $\Dd_N$ is $p^{3r-2}$.
Every instance in $\Dd_Y$ has $2r-2=k$ generators while every instance in $\Dd_N$ needs at least $2r-1=k+1$ elements to be generated.
Moreover, from Lemma~\ref{lemma:Ivanyos}, every instance in $\Dd_N$ is $1/23$-far from groups of $k$ generators.
The part of Theorem~\ref{theorem:lowerbound-k} for $k$ even then follows from the following proposition.
\begin{proposition}\label{proposition:k-even}
Any deterministic algorithm that decides with probability larger than $2/3$ whether the input is from the distribution $\Dd_Y$ or from the distribution $\Dd_N$ must use
$\RomOmega(\sqrt{p^{r-1}})$ queries.
\end{proposition}
\begin{proof}
Let us consider the decision tree associated with a deterministic algorithm $\Aa$ using $t$ queries. 
As in Section \ref{sec:unknown},  we rely on the fact that the distribution of instances generated by $\Dd$ 
can be created through a more convenient ``on the fly'' construction of $\pi$ using a random sequence $\ell$ 
of strings. We suppose hereafter that $\ell$ is fixed and denote by $\Cc_Y^{\ell}$ (resp., $\Cc_N^{\ell}$) the 
associated construction of positive (resp., negative) instances. 
We assume again that, when $\Aa$ goes through an edge corresponding to a string already seen during the computation, 
then $\Aa$ immediately stops and outputs the correct answer (this modification only improves the ability of $\Aa$). 

We denote again by $v_1,\ldots,v_n$ the set of nodes in the reduced decision tree associated with $\ell$,
and by $S\subseteq \{1,\ldots,n\}$ (resp.,~$T\subseteq \{1,\ldots,n\}$) the set of indexes~$i$ such that
 $v_i$ is a query to $O_2$ (resp.,~$O_1$).
Notice that $|S|+|T|\le t$.
For each~$j \in T$, we set $\alpha_j$ as a random variable representing the element obtained by performing a query to $O_1$.
The  answer to a query $v_i$ for $i\in S$ can be expressed as $\pi(a_i)$ where $a_i=\sum_{j\in T} k_j^i\alpha_j$ is a linear combination of the variables~$\{\alpha_j\}_{j\in T}$.
We define the function $a_{ii'}=a_i-a_{i'}=\sum_{j\in T} (k_j^i-k_j^{i'})\alpha_j$ for every $i\neq i'\in S$.
Remember that, for any positive integer $m$, we say that $a_{ii'}$ is constantly zero modulo $m$ 
if $m$ divides $k_j^i-k_j^{i'}$ for all indexes $j\in T$.
Note that we can suppose without loss of generality that for all indexes $i\neq i'\in S$ the function $a_{ii'}$ is not constantly zero modulo $p^2$
(otherwise it would give no useful information since $p^2x=0$ for any element $x$ in an instance created by $\Cc_Y^\ell$ or $\Cc_N^\ell$).

Suppose that the leaf of the reduced decision tree associated with $\ell$ corresponds to a YES decision.
The success probability of the algorithm $\Aa$ for this fixed sequence $\ell$ is at most 
\begin{equation*}
\frac{1}{2}(\rho^\ell_Y\cdot 1 +(1-\rho^\ell_Y)\cdot 1)+\frac{1}{2}(\rho^\ell_N\cdot 1+(1-\rho_N^\ell)\cdot 0)=\frac{1}{2}(1+\rho^\ell_N),
\end{equation*}
where $\rho_Y^{\ell}$ (resp., $\rho_N^{\ell}$) is the probability that $\Aa$ does not reach the leaf conditioned on the event that the instance is from $\Cc_Y^{\ell}$ (resp., from 
$\Cc_N^{\ell}$).
When the leaf of the reduced decision tree corresponds to a NO decision,
a similar calculation gives that the success probability is at most $\frac{1}{2}(1+\rho^\ell_Y)$.
Notice that $\rho^\ell_Y$ and $\rho_N^{\ell}$ are the probabilities that the same string is seen twice during the computation.
We will now show that, when the instance is created by either $\Cc_Y^{\ell}$ or $\Cc_N^{\ell}$, the inequality 
$$\Pr_{\{\alpha_j\}_{j\in T}}\left[\exists i\neq i'\in S\textrm{ such that } \sum_{j\in T} k_j^{ii'}\alpha_j=0\right]\le  \frac{t(t-1)}{2\cdot p^{r-1}}$$
holds. This implies that $\max(\rho_Y^\ell,\rho_N^\ell) \le \frac{t(t-1)}{2\cdot p^{r-1}}$ and 
then the algorithm $\Aa$ cannot distinguish $\Cc_Y^\ell$ from $\Cc_N^\ell$ with probability at least 2/3
unless $t=\RomOmega(\sqrt{p^{r-1}})$. 

Let us fix some pair of indexes $i\neq i'\in S$. If there exists some index $j\in T$ such that $k_j^{ii'}\not \equiv 0 \pmod{p}$, 
then for instances generated by $\Cc_Y^{\ell}$ and $\Cc_N^{\ell}$ we have
\begin{equation}\label{equation:even1}
  \Pr_{\{\alpha_j\}_{j\in T}}\left[\sum_{j\in T} k_j^{ii'}\alpha_j = 0\right]=\frac{1}{p^{3r-2}}.
\end{equation}
Now suppose that  $k_j^{ii'}\equiv 0 \pmod{p}$ for all $j\in T$. 
Since there are $p^{2r-2}$ elements of order at most $p$ in $\Int_{p^2}^r\times \Int_p^{r-2}$, 
and $p^{2r-1}$ elements of order at most $p$ in $\Int_{p^2}^{r-1}\times \Int_p^{r}$, 
for instances generated by $\Cc_Y^\ell$ and$\Cc_N^\ell$ we have 
\begin{equation}\label{equation:even2}
  \Pr_{\{\alpha_j\}_{j\in T}}\left[\sum_{j\in T} k_j^{ii'}\alpha_j = 0\right]\le\frac{p^{2r-1}}{p^{3r-2}} = \frac{1}{p^{r-1}}.
\end{equation}
The union bound then implies that 
$$\Pr_{\{\alpha_j\}_{j\in T}}\left[\exists i\neq i'\in S\textrm{ such that } \sum_{j\in T} k_j^{ii'}\alpha_j=0\right]\le  \frac{t(t-1)}{2\cdot p^{r-1}}$$
in both cases.

Since the same argument holds for any sequence $\ell$,
we conclude that 
the algorithm $\Aa$ cannot distinguish $\Dd_Y$ from $\Dd_N$ with overall success probability at least 2/3
unless $t=\RomOmega(\sqrt{p^{r-1}})$. 
\end{proof}

We now consider the case where $k$ is odd. Let us fix $r\geq 2$ and denote $k=2r-1$.
We define similarly $\Dd'_Y$ as the uniform distribution over binary structures for the group $\Int_{p^2}^r\times \Int_p^{r-1}$,
and $\Dd'_N$ as the uniform distribution over binary structures for $\Int_{p^2}^{r-1}\times \Int_p^{r+1}$.
The order of every instance in $\Dd'_Y$ and $\Dd'_N$ is $p^{3r-1}$.
Every instance in $\Dd'_Y$ has $2r-1=k$ generators while every instance in $\Dd'_N$ needs at least $2r=k+1$ elements to be generated.
From Lemma~\ref{lemma:Ivanyos}, every instance in $\Dd'_N$ is $1/23$-far from abelian groups generated by $k$ generators.
The part of Theorem~\ref{theorem:lowerbound-k} for $k$ odd follows from the following proposition.
\begin{proposition}\label{proposition:k-odd}
Any deterministic algorithm that decides with probability larger than $2/3$ whether the input is from the distribution $\Dd'_Y$ or from the distribution $\Dd'_N$ must use
$\RomOmega(\sqrt{p^{r-1}})$ queries.
\end{proposition}

\begin{proof}
The proof is exactly the same as the proof of Proposition~\ref{proposition:k-even},
except that Equality~(\ref{equation:even1}) becomes
$$
\Pr_{\{\alpha_j\}_{j\in T}}\left[\sum_{j\in T} k_j^{ii'}\alpha_j=0\right]=\frac{1}{p^{3r-1}}$$
and Inequality~(\ref{equation:even2}) becomes 
$$
  \Pr_{\{\alpha_j\}_{j\in T}}\left[\sum_{j\in T} k_j^{ii'}\alpha_j=0\right]\le\frac{p^{2r}}{p^{3r-1}} = \frac{1}{p^{r-1}}.\:\:
$$
\end{proof}

\subsection*{C. Proof of Theorem~\ref{theorem:known-upperbound-cyclic}}
The proof of Theorem~\ref{theorem:known-upperbound-cyclic} relies on the following theorem.

\begin{theorem}\label{theorem:close}
Let $(\Gamma,\circ)$ be a magma and let $\eta$ be a constant such that $\eta<1/120$.  
Let $G$ be a (not necessary abelian) group with order $|G|=|\Gamma|$ in 
which the multiplication of two elements $x,y$ is denoted by $xy$.
Let $f$ denote a map from $G$ to $\Gamma$.
Suppose that the following two conditions are satisfied:
\begin{itemize}
\item[(a)]
$\Pr_{x,y\in G}[f(xy)=f(x)\circ f(y)]\ge 1-\eta$;
\item[(b)]
for any subgroup $H\neq\{e\}$ of $G$ there exist two distinct elements $x,y\in H$ such that the inequality $\Pr_{u\in G}[f(xu)=f(yu)]\le 1/2$ holds.
\end{itemize}
Then there exists a binary operation $\ast\colon\Gamma\times\Gamma\to\Gamma$ such that
$(\Gamma,\ast)$ is a group isomorphic to $G$ and such that $\dist_{\Gamma}(\circ,\ast)\le 46\eta|G|^2$.
\end{theorem}

We need an auxiliary lemma to prove Theorem~\ref{theorem:close}.

Suppose that $(\Gamma,\circ)$ is a magma, $\eta$ is a constant such that $0\le \eta<1/120$,  
$G$ is a (not necessary abelian) group, and $f$ is a map from $G$ to $\Gamma$. The order of $G$ does not matter for now. 
The multiplication of two elements $x,y\in G$ is denoted by $xy$.
Following definitions introduced in the work by Friedl et al.~\cite{Friedl+STOC05}, we say that an element $x$ of $G$ is  \emph{well-behaving} if both the two inequalities $\Pr_{u\in G}[f(xu)=f(x)\circ f(u)]\ge 4/5$ and
$\Pr_{u\in G}[(f(x)\circ f(u))\circ f(u^{-1})=f(x)]\ge 4/5$ hold. Friedl et al.~showed the following results.
\begin{lemma}[Lemmas 1-6 of~\cite{Friedl+STOC05}]\label{lemma:subgroupK}
Suppose that 
\begin{equation}\label{cond1}
\Pr_{x,y\in G}[f(xy)=f(x)\circ f(y)]\ge 1-\eta.
\end{equation}
Then $\Pr_{x\in G}[x \textrm{ is not well-behaving}]\le 15\eta$. Moreover,
there exists a normal subgroup $K$ of $G$ such that, for any $x,y\in G$:
\begin{itemize}
\setlength{\itemsep}{0mm}
\item[(i)]
if $Kx=Ky$ then $\Pr_{u\in G}[f(xu)=f(yu)]\ge 1-4\eta$; 
\item[(ii)]
if $Kx\neq Ky$ then $\Pr_{u\in G}[f(xu)=f(yu)]\le 4\eta$; 
\item[(iii)]
$f(x)\neq f(y)$ for any two well-behaving elements $x$ and $y$ of $G$ such that $Kx\neq Ky$.
\end{itemize}
\end{lemma}

We now give the proof of Theorem~\ref{theorem:close}. The idea is similar to the one used in the proof of Theorem~2 in~Ref.~\cite{Friedl+STOC05}.
\begin{proof}[Proof of Theorem~\ref{theorem:close}]
Suppose that all the conditions of Theorem~\ref{theorem:close} are satisfied.
We explicitly construct a binary operation~$\ast\colon\Gamma\times\Gamma\to\Gamma$ such that  $(\Gamma,\ast)$ is isomorphic to $G$ and 
such that the Hamming distance between $(\Gamma,\circ)$ and $(\tilde\Gamma,\ast)$ is at most $46\eta|G|^2$.

Let $K$ denote the subgroup of $G$ whose existence is ensured by 
Lemma \ref{lemma:subgroupK}. From the properties of $K$ stated in 
Lemma \ref{lemma:subgroupK}, and from Condition (b)
in the statement of Theorem~\ref{theorem:close},
we conclude that $K=\{e\}$.

Let $\Gamma_1=\{f(x) \:|\; x \textrm{ is a well-behaving element of }G\}\subseteq \Gamma$
and define $\Gamma_2=\Gamma\backslash \Gamma_1$. Notice that 
$|\Gamma_1|$ is equal to the number of well-behaving elements of $G$ from Lemma~\ref{lemma:subgroupK}.

We now define a one-one map $\tilde{f}\colon G\to \Gamma$ as follows. If $x\in G$ is well-behaving, then 
$\tilde f(x)=f(x)$; if $x\in G$ is not well-behaving then $\tilde f(x)$ is an element in $\Gamma_2$ chosen arbitrarily in a way such as 
$\tilde f(x)\neq \tilde f(y)$ for distinct not well-behaving elements $x,y$ of $G$.

We define the multiplication $\ast$ over $\Gamma$ as follows. 
For any $\alpha,\beta\in\Gamma$, there exist (unique) $x_\alpha$ and $x_\beta$ in $G$ such that $\alpha=\tilde f (x_\alpha)$
and $\beta=\tilde f (x_\beta)$. We then set $\alpha\ast\beta=\tilde{f}(x_\alpha x_\beta).$
With this definition, the map $\tilde{f}$ becomes an isomorphism from $G$ to $(\Gamma,\ast)$.

We now show the following inequality:
\begin{equation}\label{property1}
\Pr_{x,y\in G}[\tilde f(x) \ast\tilde f(y)\neq \tilde f(x) \circ \tilde f(y)]\le 46\eta.
\end{equation}
By definition of $\ast$, we have $\tilde f(x) \ast\tilde f(y)=\tilde f (xy)$. With probability at least $1-45\eta$ the three elements 
$x$, $y$, and $xy$ are well-behaving elements (from Lemma~\ref{lemma:subgroupK}), in which case $\tilde f(x)=f(x)$,
$\tilde f(y)=f(y)$, and $\tilde f (xy)=f(xy)$. Remember that we also know that 
with probability at least $1-\eta$ the equality $f(xy)=f(x)\circ f(y)$ holds. Then the equality $\tilde f(x) \ast\tilde f(y) = f(x) \circ f(y)$
holds with probability at least $1-46\eta$.

Since $\tilde f$ is one-one from $G$ to $\Gamma$, Inequality~\eqref{property1} implies that $\dist_{\Gamma}(\circ,\ast)\le 46\eta|\Gamma|^2$.
\end{proof}

We are now ready to give the proof of Theorem~\ref{theorem:known-upperbound-cyclic}.
\begin{proof}[Proof of Theorem~\ref{theorem:known-upperbound-cyclic}]
Since any $\epsilon$-tester  with respect  to the Hamming distance is also an $\epsilon$-tester with respect to the edit distance,
we consider hereafter the Hamming distance. 

Suppose that the input $(\Gamma,\circ)$ is a cyclic group of order $m$. Suppose that the element $\gamma$ chosen at Step~4 is a generator of 
$(\Gamma,\circ)$. Then $\Pr_{x,y\in C_m}[f_\gamma(x+y)=f_\gamma(x)\circ f_\gamma(y)]=1$ and $\Pr_{u\in C_m}[f_\gamma(x+u)=f_\gamma(y+u)]=0$
for any $i\in\{1,\ldots,r\}$ and any distinct $x,y\in C_{m,i}$. Thus the value of the variable $\id{decision}$ at the end of the loop 
of Steps 3-13 for this specific value of $\gamma$
will always be $\const{PASS}$. Since with probability $\RomOmega(1/\log\log m)$ an element chosen uniformly at random in a cyclic group of order $m$ is a generator (see for example Ref.~\cite{Bach+96}),
by taking an appropriate value $d_1=\RomTheta(\log\log m)$ the algorithm outputs $\const{PASS}$ with probability at least 2/3.

Now suppose that $(\Gamma,\circ)$ is $\epsilon$-far from the class of cyclic groups 
and let $\gamma$ be any element of $\Gamma$. 
Denote $\tilde \epsilon=\min(\epsilon,46/120)$ and suppose that the following 
two assertions hold:
\begin{itemize}
\item[(i)]
$\Pr_{x,y\in C_m}[f_\gamma(x+y)=f_\gamma(x)\circ f_\gamma(y)]\ge 1-\tilde \epsilon/46$;
\item[(ii)]
for each index $i\in\{1,\ldots,r\}$, there exist two distinct elements $x,y\in C_{m,i}$ such that \\
$\Pr_{u\in C_m}[f_\gamma(x+u)= f_\gamma(y+u)]\le \frac{1}{2}$.
\end{itemize}

\noindent Notice that any nontrivial subgroup $H$ of $C_m$ contains at least one of the subgroups $C_{m,1},\ldots,C_{m,r}$.
Then Theorem~\ref{theorem:close} implies that $(\Gamma,\circ)$ is $\tilde \epsilon$-close (and thus $\epsilon$-close) 
to the class of cyclic groups, which contradicts our hypothesis.

We conclude that, when $(\Gamma,\circ)$ is $\epsilon$-far from the class of cyclic groups, for each value $\gamma$ chosen by the algorithm at Step~4,
at least one among Assertion (i) or Assertion (ii) should not hold.
If Assertion (i) does not hold for a specific value $\gamma$, then this is detected with probability at least $1-(1-\tilde\epsilon/46)^{d_2}$ in the tests performed at Steps 5-7.
If Assertion (ii) does not hold for a specific value $\gamma$, then there exists a value $i_0\in\{1,\ldots,r\}$ such that 
$\Pr_{u\in C_m}[f_\gamma(x+u)= f_\gamma(y+u)]\ge \frac{1}{2}$ for all distinct $x,y\in C_{m,i_0}$.
This is detected with probability at least $1-(1/2)^{d_3}$ in the tests performed at Steps 8-12.
By taking appropriate values $d_2=\RomTheta({\tilde\epsilon^{-1}}\log d_1)=\RomTheta(\epsilon^{-1}\log\log\log m)$ and $d_3=\RomTheta(\log d_1)=\RomTheta(\log\log\log m)$,
the fact that Assertion (i) or Assertion (ii) not hold will be detected with overall probability at least $2/3$ for all the values of $\gamma$ chosen by the algorithm.
Algorithm $\proc{CyclicTest}_\epsilon$ then outputs $\const{FAIL}$ with probability at least~$2/3$.

The query complexity follows from the fact that $f_\gamma$ can be evaluated using $O(\log m)$ queries and from the observation that $r=O(\log m/\log\log m)$
since an integer $n$ has at most $O(\log n/\log\log n)$ distinct prime divisors (see for example Ref.~\cite{Bach+96}).
The time complexity follows from the fact that, additionally, elements of $\Gamma$ are represented by strings of length $\ceil{\log_2 q}$.
\end{proof}

\end{document}